\documentclass[11pt]{article}
\usepackage{graphicx} 
\usepackage[utf8]{inputenc}
\usepackage[margin=1in]{geometry}
\usepackage{amsmath,amsthm,amssymb}
\usepackage{xcolor}
\usepackage{hyperref}
\usepackage{float}
\usepackage{subcaption}
\usepackage[nameinlink,capitalise]{cleveref}
\usepackage{natbib}
\usepackage{bbold}

\usepackage{algorithm}
\usepackage{algorithmic}

\usepackage{amsmath,amsthm,amssymb}
\usepackage{bbold}
\usepackage{tikz}
\usepackage{subcaption}

\makeatletter
\def\@fnsymbol#1{\ensuremath{\ifcase#1\or \dagger\or \ddagger\or
   \mathsection\or \mathparagraph\or \|\or **\or \dagger\dagger
   \or \ddagger\ddagger \else\@ctrerr\fi}}
\makeatother
\newcommand*\samethanks[1][\value{footnote}]{\footnotemark[#1]}

\title{Fair Division via the Cake-Cutting Share}
\author {
    Yannan Bai\thanks{Computer Science Department, Duke University, Durham NC 27708-0129. Email: {\tt \{yannan.bai,munagala,yiheng.shen,ian.zhang\}@duke.edu}. This work is supported by NSF awards CCF-2113798 and IIS-2402823.} \and
    Kamesh Munagala\samethanks[1] \and
    Yiheng Shen\samethanks[1] \and
    Ian Zhang\samethanks[1]
}
\date{}

\newcommand{\E}{\mathbb{E}}

\newcommand{\np}{\mathcal{P}}
\newcommand{\na}{\mathcal{A}}
\renewcommand{\ne}{\mathcal{E}}

\renewcommand{\ni}{\mathcal{I}}

\newtheorem{theorem}{Theorem}
\newtheorem{lemma}[theorem]{Lemma}

\newtheorem{definition}[theorem]{Definition}

\newcommand{\fCCS}{\mathcal{F}^{(\CCS)}} 
\newcommand{\fEFS}{\mathcal{F}^{(\EFS)}} 
\newcommand{\fEF}{\mathcal{F}^{(\EF)}} 
\newcommand{\pfEFS}{\mathcal{F}^{(\pEFS)}} 
\newcommand{\eEFS}{\mathcal{E}^{(\EFS)}} 
\newcommand{\CCS}{\mathrm{CCS}}
\newcommand{\EFS}{\mathrm{EFS}}
\newcommand{\pEFS}{\mathrm{EFS}^{\Delta}}
\newcommand{\EF}{\mathrm{EF}}
\newcommand{\PROP}{\mathrm{PROP}}

\begin{document}
\maketitle

\begin{abstract}
In this paper, we consider the classic fair division problem of allocating $m$ divisible items to $n$ agents with linear valuations over the items. We define novel notions of fair shares from the perspective of individual agents via the cake-cutting process. These shares generalize the notion of proportionality by taking into account the valuations of other agents via constraints capturing envy.  We study what fraction (approximation) of these shares are achievable in the worst case, and present tight and non-trivial approximation bounds as a function of $n$ and $m$. In particular, we show a tight approximation bound of $\Theta(\sqrt{n})$ for various notions of such shares. We show this bound via a novel application of dual fitting, which may be of independent interest. We also present a bound of $O(m^{2/3})$ for a strict notion of share, with an almost matching lower bound. We further develop weaker notions of shares whose approximation bounds interpolate smoothly between proportionality and the shares described above. We finally present empirical results showing that our definitions lead to more reasonable shares than the standard fair share notion of proportionality.
\end{abstract}

\section{Introduction}
Fair division of divisible resources among participants is a classic problem, dating back to \citet{Steinhaus} from the 1940s. In this problem, there are $m$ divisible resources, and $n$ agents with linear utility functions over these resources. The goal is to allocate the resources to the agents in a fair fashion. Several notions of fairness have been proposed in the literature, the very first of which is {\em proportionality}: Consider the equitable split, where each resource is split equally among all agents. Then in any fair division, each agent should get at least as much utility as in an equitable split. In a sense, proportionality is the ``bare minimum'' notion of fairness, and typical fair allocations guarantee other properties in addition to proportionality.

A different view of proportionality is via  ``cake cutting''~\cite{Steinhaus}, made explicit by \citet{Budish}: Imagine a participant $i$ does not know the utility functions of other participants, and needs to split the resources into $n$ bundles.
After agent $i$ has split the resources, the other participants go first and choose one bundle each, with participant $i$ getting the last pick. Suppose participant $i$ wishes to maximize its utility in the worst case over the utility functions of the other agents. Then its optimal strategy is indeed to construct each bundle as having $1/n$ share of each resource. The max-min value of this game therefore defines a fair share for participant $i$, and coincides with its utility in a proportional allocation. 

A desirable property of the proportional share is that it is not only achievable by itself, but also achievable in combination with other fairness properties such as envy-freeness (see the ``Nash welfare'' allocation~\cite{Nash,varian}). On the other hand, the value of the fair share can be very small. As an example, there are $n$ agents and $m = n$ resources, where agent $i$ has non-zero utility only for resource $i$. The proportional allocation only assigns $1/n$ fraction of resource $i$ to agent $i$, while any ``reasonable'' allocation should allocate all of resource $i$ to agent $i$. Therefore, the proportional share is a factor $n$ smaller than what one would imagine would be a reasonable fair share. 

\subsection{Cake-cutting Share} \label{sec:ccs_intro}
The above problem arises because a participant is computing the max-min shares without knowledge of other participants' utility functions. In this paper, we seek to define such max-min shares that take into account the utility functions of other participants. We do so via a slight modification of the game described above. As before, participant $i$ splits the resources into $n$ bundles, and as before, the other agents go first in choosing their bundles. However, agent $i$ now knows the utility functions of the other agents, and creates the bundles so that every other agent $j$ likes agent $i$'s bundle at most as much as any of the remaining $n-1$ bundles. This ensures that regardless of how the other agents go first and choose bundles, agent $i$'s bundle is not chosen by them. As before, participant $i$ seeks the partition that maximizes the utility they gets from its own bundle, and this utility represents its {\em cake-cutting share}. 

Note that in the above example where the agents have value for disjoint resources, agent $i$ can create $n$ bundles, where it puts the $i^{th}$ resource in a special bundle by itself, and splits the remaining resources equally among the remaining $n-1$ bundles. Any other agent strictly prefers any of the latter bundles to the special bundle, so the cake-cutting share of agent $i$ is equal to its utility for resource $i$. On this instance, the cake-cutting share obtains a factor $n$ larger utility for each agent than the proportional share.

It is relatively easy to show that the cake-cutting share of agent $i$ is equivalent to the following natural process: Agent $i$ finds a bundle $A_i$ that achieves the largest utility for itself, subject to the condition that any other agent weakly prefers the proportional allocation to bundle $A_i$ (see Lemma~\ref{lem:ccs_compute}).  Since one such bundle is the proportional allocation, the cake-cutting share is at least the proportional share for all agents. In that sense, the cake-cutting share is a {\em relaxation} of the proportional share.

\subsection{Envy-free Polyhedra and Envy-free Share} \label{sec:efs_intro}
We can generalize the game-theoretic view of fair shares further and define a class of fair shares as follows. As before, participant $i$ splits the resources into $n$ bundles. However, in addition to splitting the items, agent $i$ also {\em assigns} the bundles to the agents (including $i$ itself). The assignment can be characterized by a partition $\na = (A_1, A_2, \ldots, A_n)$, where $A_i$ represents the bundle assigned to agent $i$. For each resource, the total amount of the resource is at most its supply. We require $i$ to split the items in an envy-free way, so that $\na$ lies in the ``envy-free polyhedron'' (denoted by $\mathcal{E}_j$) of any other agent $j \neq i$. Different definitions of shares follow from different settings of the envy-free polyhedra.

In the cake-cutting share described before, any other agent $j$ values any other bundle no less than agent $i$'s bundle. This corresponds to the envy-free polyhedron $\mathcal{E}_j^{(\CCS)} = \{\na \mid u_j(A_i) \le \min_{j' \in [n]} u_j(A_{j'}) \}$, where $u_j(A)$ is the linear utility of agent $j$ for bundle $A$. Note that stricter limitations on the envy-free polyhedrons will generally lead to smaller share values. 

We also introduce a weaker instantiation of the envy-free polyhedron: $\ne_{j}^{(\EFS)} = \{\na \mid u_j(A_i) \le u_j(A_j)\}$.
This polyhedron corresponds to the allocations where  agent $j$ has no envy specifically towards agent $i$. The maximum utility that $i$ can generate for itself from any valid assignment $\na$ defines its \emph{envy-free share} ($\EFS$).

Since $\mathcal{E}_j^{(\CCS)} \subseteq \mathcal{E}_j^{(\EFS)}$ (i.e. the polyhedra for the envy-free share are a superset of those for the cake-cutting share), the envy-free share is at least the cake-cutting share. In addition, the envy-free share addresses an issue with the cake-cutting share. Consider the example of agents having values for disjoint resources from above. In this example, agent $i$'s cake-cutting share is its value for the entire resource, and hence a factor $n$ larger than the proportional share. However, if we create a copy of agent $i$, then $i$'s bundle must give this copy at most the proportional utility (see Lemma~\ref{lem:ccs_compute}), which means the cake-cutting share of $i$ drops by a factor of $n$. The envy-free share fixes this aspect -- agent $i$ can create a bundle where it assigns half the resource to its copy, so that its envy-free share only drops by a factor of $2$ when the copy is introduced. 

\subsection{Intermediate Notions of Fair Share}
Two intermediate notions of fair shares arise naturally from the above definitions.
The first notion interpolates between $\CCS$ and $\EFS$. Here, we constrain the partition $\na$ to be an \emph{envy-free} allocation across agents. In other words, all the agents will accept agent $i$'s partition $\na$ if and only if they do not envy each other (including agent $i$). Formally, the envy-free polyhedron of agent $j$ is 
$\ne_j^{(\EF)} = \{\na \mid u_j(A_j) \ge u_{j}(A_{j'}), \forall \, j' \neq j \}$. We show in Lemma~\ref{lem:order} that each agent's share for the above polyhedron is at least the cake-cutting share and at most the envy-free share, and that all these shares are at least the proportional share.

A second and more interesting intermediate notion interpolates between proportionality and $\EFS$. In $\pEFS$, when computing its $\EFS$ share, each agent $i$ knows the utility functions of all other agents except a uniformly randomly chosen set $S_{\Delta}$ of $\lfloor \frac{n - 1}{\Delta} \rfloor$ agents. For this set $S_{\Delta}$ of agents, it assumes the utility function is adversarial, so agent $i$ assigns the same bundle to these agents as it does to itself to ensure that even in the worst case over the choice of utility function, these agents do not envy agent $i$.
An agent's $\pEFS$ share is the expected value of its $\EFS$ share under this process, where the expectation is over the random choice of $S_{\Delta}$.
When $\Delta \ge n$, this reduces to $\EFS$, and when $\Delta = 1$, this reduces to proportionality. This notion therefore smoothly interpolates between proportionality at one extreme and $\EFS$ on another.

\subsection{Approximating Fair Shares}
Our goal in this paper is to study the cake-cutting and envy-free shares as novel fair share concepts. The key question we seek to answer is whether for either the cake-cutting or the envy-free shares, there is a fair division where the fair share value is simultaneously achievable for all agents, and if not, what factor of the agents' fair share is simultaneously achievable via a fair division. We study this via the notion of {\em approximation} -- in the worst case over all problem instances, what factor of this share is achievable simultaneously for all agents? 
In other words, we say that an $\alpha$ approximation is achievable for $\alpha \ge 1$ if there is a partition of the resources into $n$ bundles (one for each agent) such that each agent $i$'s utility from its bundle is at least a $1 / \alpha$ factor of its fair share. We define $\alpha(n,m)$ as the smallest achievable $\alpha$ for {\em all} instances with $m$ resources and $n$ agents.


Note that for any reasonable notion of fair share, an agent’s share value is at most its utility for the entire set of items.
Thus we must have $\alpha(n,m) \le n$ since the proportional allocation guarantees each agent a $1/n$ fraction of its total utility for the set of items.
Therefore, an approximation factor $\alpha(n,m)$ is {\em non-trivial} only if $\alpha(n,m) = o(n)$. Note that the proportional share itself satisfies $\alpha(n,m) = 1$, and is hence non-trivial. However, as discussed above, this share may be too small in many instances. We ask: {\em Does the cake-cutting share or the envy-free share also admit a non-trivial approximation factor?} 

As an aside, we note that for any problem instance, the optimal approximation to the cake-cutting (resp. envy-free) share for this instance is computable via linear programming. (See Section~\ref{sec:efs_def}.) 
Our goal is different -- we wish to find a uniform bound on $\alpha(n,m)$ as a function of $m$ and $n$.

\subsection{Our Results}
Our main result is to show that the above defined shares are interesting and non-trivial. We show this by demonstrating non-trivial bounds on  $\alpha(n,m)$,  meaning that all agents can indeed {\em simultaneously} achieve a non-trivial (that is $\omega(1/n)$) fraction of their cake-cutting (resp. envy-free) share. 

We first show the following theorem, where we use $\alpha(n, \cdot)$ to denote the worst case approximation over all instances with $n$ agents, and an arbitrary number of items:


\begin{theorem}[Proved in Section~\ref{sec:main1}]
\label{thm:main1}
For the envy-free share (and hence the cake-cutting share), $\alpha(n,\cdot) =  O(\sqrt{n})$.  
\end{theorem}


In other words, for {\em any} instance with $n$ agents and $m$ items, it is possible to find an allocation of the resources to the agents where each agent obtains utility  $\Omega\left(\frac{1}{\sqrt{n}}\right)$ times its cake-cutting (resp. envy-free) share. 
We next show there are instances where this bound is tight, so that it is not possible for all agents to achieve a better fraction of their cake-cutting share. 

\begin{theorem}[Proved in Section~\ref{sec:lb1}]
\label{thm:main11}
For the cake-cutting share (and hence the envy-free share), $\alpha(n,\cdot) =  \Omega(\sqrt{n})$.  
\end{theorem}

We note that the above results imply the same bounds for any intermediate notion of fair shares where the constraints in the envy polyhedron are stricter than those for the envy-free share and weaker than those for the cake-cutting share.

We next show a similar non-trivial (that is $o(m)$) approximation for the cake cutting share as a function of $m$, the number of resources, again using the notation $\alpha(\cdot, m)$ to denote the worst case approximation for instances with $m$ items and arbitrary number of agents.

\begin{theorem}[Proved in Section~\ref{sec:main2}]
\label{thm:main2}
For the cake-cutting share, $\alpha(\cdot,m) = O(m^{2/3})$. Further, there exists an instance of the fair division problem where $\alpha( \cdot,m) = \Omega(\sqrt{m})$.
\end{theorem}

We finally extend the proof of Theorem~\ref{thm:main1} to the $\pEFS$ share described above, where agent $i$ does not know the utility function of a randomly chosen set $S_{\Delta}$ of agents of size $\lfloor \frac{n-1}{\Delta} \rfloor$, for whom it replicates its own allocation. We show the following theorem where the approximation ratio increases smoothly in the parameter $\Delta$. 

\begin{theorem}[Proved in Section~\ref{app:partial}]
\label{thm:main3}
For the $\pEFS$ share, $\alpha(n,\cdot) =  \Theta(\sqrt{\Delta})$. The upper bound of $O(\sqrt{\Delta})$ holds for {\em any} set $S_{\Delta}$ that $i$ chooses, while in the lower bound instance, the $\Omega(\sqrt{\Delta})$ bound holds for {\em all} possible sets $S_{\Delta}$.
\end{theorem}

Note that in agent $i$'s partition, agent $i$ values its own bundle at most as much as any bundle assigned to the agents in $S_\Delta$.
Thus, any agent's $\pEFS_i$ share is at most a $1/(1 + \lvert S_\Delta \rvert)$ factor of its utility for the entire set of items. This implies that any agent's $\pEFS_i$ share is at most a $\Delta$ factor larger than its proportional share, which gives an easy upper bound of $\Delta$ for $\alpha(n, \cdot)$. The above theorem improves this to $O(\sqrt{\Delta})$.


\paragraph{Techniques.} The main technical highlight is the proof of the upper bounds on approximation. To show that a $O(\sqrt{n})$ approximation is always achievable, we use the idea of dual fitting, where we write the worst case approximation over all problem instances as a linear program with exponentially many variables. We then show a feasible solution to the dual of this program with value $O(\sqrt{n})$. This linear program is non-trivial, and the techniques used for writing it may find other applications to fair division and beyond. Indeed, we find the lower bound instances by solving this linear program for small values of $n$.  The proof of the $O(m^{2/3})$ upper bound on approximation is via a greedy set cover type algorithm, showcasing an interesting connection between the cake-cutting share and set cover. 




\paragraph{Empirical Results.} Given that various versions of the envy polyhedron yield comparable worst-case approximation bounds, a natural question is what notion would be reasonable in practice. Towards this end, we compute the approximations on simulated and real problem instances in Section~\ref{sec:expt} and show that the cake-cutting share --  the strictest envy polyhedron -- comes closest to having an approximation factor of one, meaning these shares are realizable via an allocation. We further show that the approximation factor of $\pEFS$ smoothly decreases with $\Delta$.
These results complement our theoretical bounds, and serve as another justification for developing these novel definitions of shares.

\subsection{Related Work}
\paragraph{Proportionality and Fair Allocation.}
The concept of {\em proportionality}  in fair division was first introduced by \citet{Steinhaus}. This foundational work addressed the problem of dividing a heterogeneous resource (the ``cake'') among several participants in such a way that each participant perceives their share as fair, where one notion of fairness presented was  proportionality.
A proportional allocation for divisible items is easy to achieve by evenly splitting every item; however its generalization to groups of agents has been studied for shared resources \cite{FainGM17}.

\paragraph{Fair Shares for Indivisible Goods.} 
 We refer readers to \cite{mishra2023fair,amanatidis2022fair} for detailed surveys of recent results on proportionality and other fair allocations of indivisible items. A natural approach for indivisible items is to approximate proportionality, and several works \cite{conitzer2017fair,aziz2020polynomial,baklanov2021achieving,baklanov2021propm} have considered such approximate notions.

The work of \citet{Babaioff} defines the general notion of a \emph{share} that is computed for each agent. The most popular notion of a share is the maximin share (MMS) introduced by \citet{Budish}, which calculates each agent's share using the cutting and choosing game described previously. Although allocations that achieve exact MMS may not exist for indivisible goods \cite{procaccia2014fair}, constant approximations to this share can be achieved \cite{akrami2024breaking,babaioff2022best,hosseini2022ordinal,akrami2023improving}. 
Recently, \citet{babichenko2024fair} propose the {\em quantile share}, where an agent wants its share to be better than a uniformly random allocation with constant probability. 

All the above notions compute agent $i$'s share solely as a function of the number of agents $n$, and agent $i$'s personal utility function $v_i$. As such, for divisible goods, they all reduce to proportionality. As far as we are aware, our work is the first to define and study fair shares that are cognizant of the utility functions of other agents, therefore moving beyond proportionality even for the divisible goods setting.

\paragraph{Envy-freeness and Cake Cutting}
Envy-freeness -- where no agent envies the allocation of another agent -- is more demanding than proportionality and was first formally introduced by \citet{gamow1958puzzle}. Envy-free allocations have been widely studied for both divisible and indivisible goods; see \cite{varian1973equity,varian,Caragiannis2019nash,lipton,aziz2016discrete} and the citations therein. Our proposed shares are related in spirit to  envy-free allocations, but differ in the details, since envy in our context is with regard to the agent computing its share. We note that to the best of our knowledge, our work is the first to define fair shares (or guarantees of utility) based on envy-freeness. 

\section{Preliminaries}
In an instance of the fair division problem, there is a set $N$ of $n$ agents (participants), and a set $M$ of $m$ homogeneous divisible items or resources to be allocated. Without loss of generality, we assume $ N = [n] = \{1, \ldots, n\}$ and $M = [m] = \{1, \ldots, m\}$. For each agent $i$, its valuation for a single item is described by a function $v_i \colon M \to \mathbb{R}_{\ge 0}$, where $v_{ik}$ is agent $i$'s valuation for the entirety of item $k$.
We will let $V = (v_1, \dots, v_n)$ be the vector of all valuation functions. Therefore, a problem instance in our fair-division model is represented by a triple $\ni = \langle [n], [m], V \rangle$. 

A bundle $B$ is described by a vector $(x_1, \dots, x_m) \in [0, 1]^m$, which denotes that for all $k \in [m]$, bundle $B$ contains $x_k$ quantity of item $k$. 
We write $B(k) = x_k$ to describe the quantity of item $k$ in the bundle $B$.
We will assume that the utilities are additive and each item is homogeneous, so agent $i$'s utility for a bundle $B$ is given by 
$u_i(B) = \sum_{k \in [m]} v_{ik} \cdot B(k)$.

An allocation $\na = (A_1, \dots, A_n)$ is a tuple of $n$ bundles such that at most the unit supply of each item is allocated, i.e., $\sum_{i \in [n]} A_i(k) \le 1$ for all $k \in [m]$.
Each agent $i$ receives bundle $A_i$ with utility $u_i(A_i)$.
We denote the set of all allocations of $m$ items among $n$ agents by $\np(n, m)$.

\paragraph{Cake-Cutting (CCS) and Envy-Free (EFS) Share.}
 \label{sec:efs_def}
Fix a problem instance $\ni = \langle [n], [m], V \rangle$. We first define the cake-cutting share of each agent. For any agent $i$, we use
$\fCCS_i(\ni) = \{\na \in \np(n, m) \mid u_j(A_i) \le u_j(A_{j'}), \  \forall j, j' \in [n] \setminus \{i\}\}$
to denote the set of all feasible allocations by agent $i$, such that every other agent $j \in [n] \setminus \{i\}$ weakly prefers any other bundle $A_{j'}$ to bundle $A_i$. Recall that in Section~\ref{sec:efs_intro}, we defined the envy-free polyhedron of agent $j$ in the cake-cutting share to be $\mathcal{E}_j^{(\CCS)}(\ni) = \{\na \in \np(n, m) \mid u_j(A_i) \le \min_{j' \in [n]} u_j(A_{j'}) \}.$ One can verify that the space of all feasible allocations by agent $i$ is exactly the intersection of all other agents' envy-free polyhedra.
\begin{definition} [Cake-cutting Share]
    For any agent $i \in [n]$, we define the \emph{cake-cutting share} of agent $i$ to be \[\CCS_i(\ni) = \max_{(A_1, \dots, A_n) \in \fCCS_i(\ni)}u_i(A_i).\]
\end{definition}

We now formally define the envy-free share, which strengthens the cake cutting share. For any agent $i$, we use 
$\fEFS_i(\ni) = \{\na \in \np(n, m) \mid u_j(A_{i}) \le u_j(A_j), \ \forall j \in [n] \setminus \{i\}\}$
to denote the set of allocations such that every other agent $j \in ([n] \setminus \{i\})$ weakly prefers its bundle $A_{j}$ to bundle $A_i$. Again, one can verify that $\fEFS_i(\ni) = \bigcap_{j: j \neq i} \eEFS_j(\ni)$.

\begin{definition} [Envy-free Share]
\label{def:efs}
    For any agent $i \in [n]$, we define the \emph{envy-free share} of agent $i$ to be \[\EFS_i(\ni) = \max_{(A_1, \dots, A_n) \in \fEFS_i(\ni)} u_i(A_i).\]
\end{definition}

Since $\fCCS_i(\ni)$ is clearly a subset of $\fEFS_i(\ni)$ by comparing the definitions of both shares, we have $\EFS_i(\ni) \ge \CCS_i(\ni)$ for all agents $i \in [n]$. For simplicity, we will  use $\CCS_i$ and $\EFS_i$ to denote an agent's two shares respectively when it is clear which instance we are considering. 

\paragraph{Computation.} Given a problem instance as input, we now consider the computation of both shares. First, we show that the cake-cutting share $\CCS_i$ is also the maximal value of $u_i(A_i)$ such that all other agents weakly prefer the proportional allocation to bundle $A_i$.
We prove this by showing there exists an allocation $\na \in \fCCS_i(\ni)$ which maximizes $u_i(A_i)$ and allocates $1 - A_i(k)$ equally among the other agents $j \neq i$ for each item $k$.

\begin{lemma}
\label{lem:ccs_compute}
    For any agent $i$, the value $\CCS_i$ is the optimal solution to the following LP: 
    \begin{align}
        \max \qquad &\sum_{k \in [m]} v_{ik} \cdot x_k   \,,  \label{eq:LP-CCSval-obj} \\
        \textrm{s.t.} \qquad &\sum_{k \in [m]} v_{jk} \cdot x_k  \le \frac{u_j([m])}{n}\,, \ \  \forall\, j \in ([n] \setminus \{i\}) \,; \label{cons:LP-CCSval-prop}\\
        & 0 \le x_k \le 1 \,, 
        \ \ \forall\, k \in [m]  \,. \notag
    \end{align}
\end{lemma}
\begin{proof}
    We can interpret the LP in Eq~(\ref{eq:LP-CCSval-obj}) as constructing agent $i$'s bundle $A_i = (x_1, \dots, x_m)$ in an allocation $\na$ where $A_i(k) = x_k$.
    The objective function Eq~(\ref{eq:LP-CCSval-obj}) is exactly agent $i$'s utility for bundle $A_i$,
    so it remains to show that there exists an allocation $\na = (A_1, \dots, A_n) \in \fCCS_i(\ni)$ if and only if $A_i$ satisfies Eq~(\ref{cons:LP-CCSval-prop}).
    
    If $A_i$ satisfies Eq~(\ref{cons:LP-CCSval-prop}), then we can construct the other bundles by splitting the remaining items $\mathbb{1}_n - A_i$ equally among the bundles.
    In such a partition, each bundle $A_{j'} \neq A_i$ contains exactly $(1 - x_k)/(n-1)$ of each item $k$.
    By Eq~(\ref{cons:LP-CCSval-prop}), any agent $j \in ([n] \setminus \{i\})$ has utility at least $u_j([m]) - u_j([m])/n = u_j([m])(n-1)/n$ for bundle $(\mathbb{1}_n - A_i)$.
    Thus for any $j' \in ([n] \setminus \{i\})$, we have $u_j(A_{j'}) \ge u_j([m])(n-1)/(n(n-1)) = u_j([m])/n \ge u_j(A_i)$ as desired.

    Now consider any allocation $\na = (A_1, \dots, A_n) \in \fCCS_i(\ni)$.
    Assume for contradiction that Eq~(\ref{cons:LP-CCSval-prop}) does not hold, i.e., there is an agent $j \in ([n] \setminus \{i\})$, such that $j$ values bundle $A_i$ strictly more than the proportional allocation.
    Since the utilities are additive, there must exist another bundle $A_{j'}$ that agent $j$ strictly values less than the proportional allocation.
    However, this implies $u_j(A_i) > u_j(A_{j'})$, which contradicts that $\na \in \fCCS_i(\ni)$. This completes the proof.
    \end{proof}

Similarly, the envy free share $\EFS_i$ is the solution to the following linear program:

\begin{align}
        \max \quad &\sum_{k \in [m]} v_{ik} \cdot x_{ik}   \,,  \label{eq:LP-EFSval-obj} \\
        \textrm{s.t.} \quad &\sum_{k \in [m]} v_{jk} \cdot x_{ik}  \le \sum_{k \in [m]} v_{jk} \cdot x_{jk}\,, \ \ \forall\, j \in ([n] \setminus \{i\}) \,;  \label{cons:LP-EFSval-prop} \\
        & \sum_{j \in [n]} x_{jk} \le 1 \,, \ \ \forall\, k \in [m] \,; \notag \\
        &  x_{jk} \ge 0 \,, \ \  \forall\, j \in [n], \ k \in [m]  \,. \notag
\end{align}

\paragraph{Scale Invariance and Ordering.}  We now show that the values of the shares are invariant to scaling the utility functions of different agents by different factors.

\begin{lemma}
\label{lem:scale_invar}
\label{lem:scale_invar2}
    Consider an instance $\ni = \langle [n], [m], V \rangle$ of the fair division problem.
    Now create an alternative instance $\ni' = \langle [n], [m], V' \rangle$ where for every agent $i$, all of its values are scaled by $\alpha_i > 0$, i.e., $v'_{ik} = \alpha_i \cdot v_{ik}$ for all $k \in [m]$. Then $\CCS_i(\ni') = \alpha_i \cdot \CCS_i(\ni)$ and $\EFS_i(\ni') = \alpha_i \cdot \EFS_i(\ni)$.
\end{lemma}
\begin{proof}
    Consider the LP formulation of $\CCS_i$ in Lemma~(\ref{lem:ccs_compute}). For instance $\ni'$, the constraints remained unchanged and the objective function is scaled by $\alpha_i$, implying the result. The same argument applies to $\EFS_i$.
\end{proof}

\paragraph{Intermediate Fair Shares.} We note that neither the polyhedra corresponding to $\CCS_i$ nor $\EFS_i$ capture exact envy -- the former polytope captures any agent $j$ preferring $i$'s bundle the least among the $n$ bundles, while the latter captures agent $j$ preferring $i$'s bundle at most as much as $j$'s own bundle. As mentioned in the previous section, true envy-freeness would be captured by
$\fEF_i(\ni) = \{\na  \in  \np(n, m) \mid u_j(A_{j}) \ge u_j(A_{j'}), \ \forall j \in [n] \setminus \{i\}, j' \in [n] \}.$
Let $\EF_i(\ni) = \max_{(A_1, \dots, A_n) \in \fEF_i(\ni)} u_i(A_i)$ denote this share. Similarly, denote the proportional share as $\PROP_i(\ni) = u_i([m])/n$.

We now show the following ordering result of these shares, hence showing the cake cutting share imposes a stronger constraint than pure envy.

\begin{lemma} \label{lem:order}
For any instance $\ni$ and all agents $i \in [n]$:
$$ \PROP_i(\ni) \le \CCS_i(\ni) \le \EF_i(\ni) \le \EFS_i(\ni).$$
\end{lemma}
\begin{proof}
For any instance $\ni$ and any agent $i$, we have $\fCCS_i(\ni) \subseteq \fEFS_i(\ni)$ and $\fEF_i(\ni) \subseteq \fEFS_i(\ni)$. By the strictness of the constraint sets, it follows that $\CCS_i \le \EFS_i$ and $\EF_i \le \EFS_i$. Next note that the proportional allocation satisfies the LP in Lemma~\ref{lem:ccs_compute}, so that $\PROP_i \le \CCS_i$. Further, the proof of Lemma~\ref{lem:ccs_compute} shows that an optimal solution to $\CCS_i$ allocates, for each item $k$, the quantity $1 - x_k$ equally among the agents $j \neq i$. Therefore, any agent $j \neq i$ does not envy any other agent $j' \in [n]$. This shows that the solution satisfies the polytope $\fEF_i$, so that $\CCS_i \le \EF_i$. This completes the proof.
\end{proof}

\paragraph{Approximating the Shares and $\alpha(n,m)$.}
\label{sec:approx}
On any instance $\ni = \langle [n],[m],V \rangle$, we seek an allocation where the utility of each agent $i$ is at least $\CCS_i$ (resp. $\EFS_i$). Since this may not be possible, we consider the ``best possible'' approximation. On any instance, this allocation can be computed by solving the following LP, which we describe for $\CCS_i$; an analogous LP holds for $\EFS_i$:

\begin{align}
    \max \qquad & \theta \,, & \label{lp:approx}\\
    \textrm{s.t.} \qquad & \sum_{k} v_{i k} \cdot x_{i k} \ge \theta \cdot \CCS_i \,, \ \ 
    \forall \, i \in [n] \,; \notag \\
    & \sum_{i} x_{i k} \le 1 \,, \ \ 
    \forall \, k \in [m] \notag \,; \\
    & x_{ik} \ge 0 \,, \ \ \forall \, i \in [n], \ k \in [m] \,. \notag
\end{align}

Let $\theta(\ni)$ denote the optimal approximation for instance $\ni$. As shown above, this can be computed by a linear program. However, our goal is different; we seek to bound the worst-case approximation. We denote $\alpha(n,m) = \max_\ni \frac{1}{\theta(\ni)}$, where the maximum is over all instances with $n$ agents and $m$ items. This represents the worst-case approximation bound for such instances. Our goal is to show that $\alpha(n,m)$ is small, in particular that it beats the trivial bound of $\alpha(n,m) = n$ that is achievable for any fair share via proportionally allocating each item.

\paragraph{$\pEFS$ Share.} As discussed before, in this notion, for each agent $i$ there is a set $W_i$ of agents of size $\lfloor \frac{n - 1}{\Delta} \rfloor$, whose utility $i$ is unaware of. We again consider the set of feasible allocations $\fEFS_i(\ni)$. However, we constrain this set to be feasible {\em regardless} of $\{u_j\}_{j \in W_i}$. It is easy to check that this forces $A_j = A_i$ for $j \in W_i$. Therefore, we can write $\pfEFS_i(\ni,W_i) = \{\na \in \np(n, m) \mid u_j(A_{i}) \le u_j(A_j), \ \forall j \in [n] \setminus (W_i \cup \{i\}); \ A_j = A_i, \  \forall j \in W_i\}$. We can now define $\pEFS_i$ analogous to Definition~\ref{def:efs} as:

\begin{definition} [$\pEFS$ Share]
\label{def:pefs}
    For any agent $i \in [n]$,
    \[\pEFS_i(\ni,W_i) = \max_{(A_1, \dots, A_n) \in \pfEFS_i(\ni,W_i)} u_i(A_i). \]
    We define $\pEFS_i(\ni) = \E_{W_i}[\pEFS_i(\ni,W_i)]$ where the expectation is over a uniformly random  $W_i$ of size $\lfloor \frac{n - 1}{\Delta} \rfloor$.
\end{definition}

For any given $W_i$, the value $\pEFS_i(\ni,W_i)$ is the solution to an LP analogous to Eq~(\ref{eq:LP-EFSval-obj}); this is presented in Section~\ref{app:partial}. We can now efficiently estimate $\pEFS_i(\ni)$ to an arbitrary approximation by randomly sampling  $W_i$ repeatedly and solving the LP over the samples.

We can now define $\alpha(n,m)$ the same as before, and we present bounds for $\alpha(n,\cdot)$ in Section~\ref{app:partial}.



\section{Upper Bound: Proof of Theorem~\ref{thm:main1}}
\label{sec:main1}

We now show that $\alpha(n,\cdot) = O(\sqrt{n})$. By Lemma~\ref{lem:order}, it suffices to show an upper bound on approximation for the envy-free share, $\EFS_i$. Note that the same upper bound also holds for the intermediate notion of $\EF_i$.




\paragraph{Step 1: Reduction to Welfare Maximization.}
We now prove Theorem~\ref{thm:main1} in a sequence of three steps. 
We first show it suffices to bound the ratio between the social welfare, $SW(\ni) = \sum_{k \in [m]} \max_{i \in [n]} v_{ik}$, and the sum of envy-free shares, $C(\ni) = \sum_i \EFS_i$,
for all instances $\ni = \langle [n], [m], V \rangle$ with $n$ agents. 
The following theorem is a direct application of strong duality. 

\begin{theorem}
\label{thm1}
Suppose for all instances $\ni$ with $n$ agents, we have
$SW(\ni) \ge \theta_n \cdot C(\ni).$
Then for all instances with $n$ agents, there exists an allocation where each agent $i$ obtains utility at least $\theta_n \cdot \EFS_i$.
\end{theorem}

\begin{proof}
    We wish to show that the optimal solution of the following LP (LP1) has value at least $\theta_n$:
    \begin{align}
        \max \qquad & \lambda \,, \notag\\
        \textrm{s.t.} \qquad & \sum_{k} v_{i k} \cdot y_{i k} \ge \lambda \cdot \EFS_i \,, \ \  \forall \, i \in [n] \,; \label{cons:LP1_1} \\
        & \sum_{i} y_{i k} \le 1 \,, \ \ 
        \forall \, k \in [m] \label{cons:LP1_2} \,; \\
        & y_{ik} \ge 0 \,, \ \  \forall \, i \in [n], \ k \in [m] \,. \notag
    \end{align}
    Let $\theta'$ be the optimal value of LP1. We take the dual of LP1, with dual variables $\{\beta_i\}_{i \in [n]}$ (for constraints \ref{cons:LP1_1}) and $\{\gamma_k\}_{k \in [m]}$ (for constraints \ref{cons:LP1_2}):
    \begin{align}
        \min \qquad & \sum_{k} \gamma_k \,,  \notag\\
        \textrm{s.t.} \qquad & \sum_{i} \EFS_i \cdot \beta_i \ge 1 \,; \label{cons:LP2_1} \\
        & -v_{i k}\cdot \beta_i + \gamma_k \ge 0 \,, \ \  \forall \, i \in [n], \ k \in [m] \,; \label{cons:LP2_2} \\
        & \beta_i \ge 0, \gamma_k \ge 0 \,, \ \ \forall \, i \in [n], \ k \in [m] \,. \notag
    \end{align}
    By duality, there exist variables $\{\beta_i^*\}_{i \in [n]}$ and $\{\gamma_k^*\}_{k \in [m]}$ that satisfy the above constraints while $\sum_k {\gamma_k^*} = \theta'$.
    
    Now consider a scaled instance 
    $\ni'$ where the value of each agent $i$ for each item $k$ is $\beta_i^* \cdot v_{ik}$. By Lemma~\ref{lem:scale_invar}, the corresponding envy-free shares in $\ni'$ become $\beta_i^* \cdot \EFS_i$.
    Applying the precondition of the theorem, we have $SW(\ni') \ge \theta_n \cdot C(\ni')$, so there must be an allocation $\mathbf{x} = \{x_{i k}\}_{i\in [n], k \in [m]}$ that satisfies the following constraints:
    \begin{align}
        & \sum_{i, k} \beta_i^* \cdot v_{i k} \cdot x_{i k} \ge \sum_{i} \beta_i^* \cdot \EFS_i \cdot \theta_n \,; \label{eq:cond_1}\\
        & \sum_{i} x_{i k} \le 1 \,, \ \ 
        \forall \, k \in [m] \notag \,; \\
        & x_{ik} \ge 0 \,, \ \  \forall \, i \in [n], \ k \in [m] \,. \notag
    \end{align}
    Thus by Eq~(\ref{eq:cond_1}) and Eq~(\ref{cons:LP2_1}), we have 
    \begin{equation}
        \label{eq:ge_theta_n}
        \sum_{i, k} \beta_i^* \cdot v_{i k} \cdot x_{i k} \ge \sum_{i} \beta_i^* \cdot \EFS_i \cdot \theta_n \ge \theta_n. 
    \end{equation}
    Since $x_{ik} \ge 0$ for all $i$ and $k$,
    we have by Eq~(\ref{cons:LP2_2}) that
    \[
    \sum_{k} \beta_i^* \cdot v_{i k} \cdot x_{i k} \le \sum_{k} \gamma_k^* \cdot x_{i k}, \ \forall i \in [n]. 
    \]
    Summing up on $i$, we get 
    \[
    \sum_{i, k} \beta_i^* \cdot v_{i k} \cdot x_{i k} \le \sum_{i, k} \gamma_k^* \cdot x_{i k} = \sum_{k} \gamma_k^* \cdot \left(\sum_{i} x_{i k} \right) \le \sum_{k} \gamma_k^* = \theta'.
    \]
    Finally, we apply Eq~(\ref{eq:ge_theta_n}) to get $\theta' \ge \theta_n$ as desired.
\end{proof}

\paragraph{Step 2: Reduction to Binary Utilities.}
Given Theorem~\ref{thm1}, our goal is to show that the precondition holds for $\theta_n = \Omega(1/\sqrt{n})$. In other words, it suffices to show that for any instance $\ni = \langle [n], [m], V \rangle$, the sum of envy-free shares, $C(\ni)$, is at most $O(\sqrt{n})$ times the social welfare, $SW(\ni)$.  

We prove that it suffices to bound the ratio $C(\ni)/SW(\ni)$ for binary instances. An instance is said to be binary if $v_{ik} \in \{0,1\}$ for all $i \in [n], k \in [m]$. We show that:

\begin{theorem}
\label{thm2}
Suppose that for all binary instances $\ni$ with $n$ agents, we have $C(\ni) \le \alpha_n \cdot SW(\ni)$. Then for all instances with general valuations and $n$ agents, we have $C(\ni) \le \alpha_n \cdot SW(\ni)$.
\end{theorem}
\begin{proof}
Given a general instance $\ni$, we construct a binary instance with the same number of agents, where the sum of envy free shares does not decrease, while the social welfare remains the same. This will show the theorem.

Fix a small $\epsilon > 0$. First scale all values so that $\min(C(\ni), SW(\ni)) = 1$. Note that this preserves the ratio $\frac{C(\ni)}{SW(\ni)}$, since both quantities are scaled by the same amount, by Lemma~\ref{lem:scale_invar}. Next construct $\ni' = \langle [n], [m], V' \rangle$, where $v'_{ik} = \lceil \frac{v_{ik}}{\epsilon} \rceil$. Without the ceiling, this is just scaling all values by $\frac{1}{\epsilon}$, so this again preserves $\frac{C(\ni)}{SW(\ni)}$ by Lemma~\ref{lem:scale_invar}. Further, $C(\ni'), SW(\ni') \ge \frac{1}{\epsilon}$. If we take the ceiling of each value, the denominator does not decrease and the numerator increases by at most an additive $mn$, so the ratio increases by a multiplicative $(1 + O(\epsilon \cdot mn ))$ for small enough $\epsilon$. For fixed $n,m$, in the limit as $\epsilon \rightarrow 0$, we can ignore this increase in ratio. In other words, $\frac{C(\ni')}{SW(\ni')} \rightarrow \frac{C(\ni)}{SW(\ni)}$ as $\epsilon \rightarrow 0$.

We now create a binary instance  $\ni'' = \langle [n], [m'], V'' \rangle$ as follows. For item $k$, let $q_k = \max_{i \in [n]} v'_{ik}$. Note that $q_k$ is an integer. Create $q_k$ copies of this item as $k_1, k_2, \ldots, k_{q_k}$. For agent $i$, suppose $v'_{ik} = a$, we set $v''_{ik_w} = 1$ for $w \in \{1,2,\ldots,a\}$ and $v''_{ik_w} = 0$ for $w \in \{a+1, \ldots,q_k\}$.  

First note that $SW(\ni'') = SW(\ni')$. This is because item $k$'s contribution to  the RHS is $\max_{i \in [n]} v'_{ik} = q_k$, which is the number of copies of $k$ created in $\ni''$, who each contribute $1$ to the welfare.  Next, consider for each agent $i$ the allocation $x$ corresponding to the LP optimum for $\EFS_i$ in the instance $\ni'$, as written in Eq~(\ref{eq:LP-EFSval-obj}). We use this to construct a feasible solution for the same LP for instance $\ni''$ as follows: For each item $k$ in $\ni'$ and $j \in [n]$, we assign the value $x_{jk}$ to the corresponding variables $x_{jk_w}$ for each copy $k_w$ of item $k$ in the instance $\ni''$. This preserves the objective value, and the LHS and RHS of Eq~(\ref{cons:LP-EFSval-prop}) for each $j$. This means $\EFS_i$ for each agent $i$ cannot decrease in transforming $\ni'$ to $\ni''$. Therefore, $\frac{C(\ni')}{SW(\ni')} \le \frac{C(\ni'')}{SW(\ni'')}$, completing the proof.
\end{proof}

\paragraph{Step 3: Dual Fitting for Binary Utilities.}
\label{sec:step3}
We finally upper bound the ratio $\frac{C(\ni)}{SW(\ni)}$ for  binary utilities. Combined with Theorem~\ref{thm2} and Theorem~\ref{thm1}, this proves Theorem~\ref{thm:main1}.

\begin{theorem}
\label{thm3}
    For any instance of the fair division problem with $n$ agents and binary utilities, 
    $\frac{C(\ni)}{SW(\ni)} = O(\sqrt{n})$,
    where $C(\ni) = \sum_i \EFS_i$ and $SW(\ni)$ is the social welfare.
\end{theorem}

\begin{proof}
Let $\ni = \langle [n], [m], V \rangle$ be a binary instance.
We assume that every item is valued by at least one agent, since we can discard any unvalued items without affecting the problem instance.
We first scale the values down by $m$, so that the total welfare is $1$. By Lemma~\ref{lem:scale_invar}, this scales down $C(\ni)$ by the same factor. For agent $i$, let $Q_i = \{k \mid v_{ik} = 1/m\}$.

We write the ratio $\frac{C(\ni)}{SW(\ni)}$ as a linear program:
\begin{align*}
    \max \qquad & \sum_{i \in [n]} \sum_{k \in Q_i}  x_{iik} \,, \\
    \textrm{s.t.} \qquad & \sum_{k \in Q_i \cap Q_j} x_{iik} \le \sum_{k \in Q_j} x_{ijk} \,, \ \ \forall \, i \neq j \,; \\
    & \sum_{j \mid k \in Q_j} x_{ijk} \le 1 \,, \ \ \forall\, i \in [n], \ k \in [m] \,; \\
    & x_{ijk} \ge 0 \,, \ \ \forall \, i,j \in [n], \ k \in Q_j \,.
\end{align*}

In the above LP, the variables $\{x_{ijk}\}$ for any $i$ correspond to the variables from the LP for $\EFS_i$ in Eq~(\ref{eq:LP-EFSval-obj}). We simply take the intersection of the constraints for different $i$ and add up the objectives.
It suffices to only have variables $x_{ijk}$ for items $k \in Q_j$, since it is never beneficial to allocate items to agents who do not value them.

We need to find the worst case of the above LP over sets $\{Q_i\}$. Towards this end, index items by the subset $S$ of agents that have value $1/m$ for that item. Note that multiple items could correspond to the same subset. If $b$ items correspond to the same subset $S$, then we set the weight $v_S = b / m$. Since the social welfare is $1$, we have $\sum_S v_S = 1$. Let $B_i = \{S \subseteq [n] \mid i \in S\}$ denote the set of subsets that contain agent $i$,
and $x_{ijS}$ denote the variable $x_{ijk}$ corresponding to the item $k = S$ in the LP formulation above. The above LP is now the following:
\begin{align*}
    \max \qquad & \sum_{i \in [n]} \sum_{S \in B_i}  v_S \cdot x_{iiS} \,, \\
    \textrm{s.t.} \qquad & \sum_{S \in B_i \cap B_j} v_S \cdot x_{iiS} \le  \sum_{S \in B_j} v_S \cdot x_{ijS} \,,  \ \  \forall \, i \neq j \,; \\
    & \sum_{j \mid S \in B_j} x_{ijS} \le 1 \,,   \ \  \forall\, i \in [n], \ S \subseteq [n] \,; \\
    & x_{ijS} \ge 0 \,,  \ \  \forall \, i,j \in [n], \ S \subseteq [n] \,.
\end{align*}

We now need to find the worst case value of the above LP over non-negative $\{v_S\}$ that satisfy $\sum_S v_S = 1$. Replace $z_{ijS} = v_S \cdot x_{ijS}$. Substituting in the above LP with the added condition $\sum_S v_S \le 1$, and treating both $\{z_{ijS}\}$ and $\{v_S\}$ as variables, we obtain the following LP that upper bounds $\frac{C(\ni)}{SW(\ni)}$ over all binary instances with $n$ agents.

\begin{align*}
    \max \qquad & \sum_{i \in [n]} \sum_{S \in B_i} z_{iiS}   \,, \\
    \textrm{s.t.} \qquad & \sum_{S \in B_i \cap B_j} z_{iiS} \le \sum_{S \in B_j}  z_{ijS} \,,
    \ \ \forall\, i \neq j  \,; \\
    &\sum_{j \mid S \in B_j} z_{ijS} \le v_S \,, \ \ \forall\, i \in [n], \ S \subseteq [n] \,;\\
    & \sum_{S} v_S \le 1 \,; \\
    & z_{ijS}, v_{S} \ge 0 \,, \ \  \forall \, i,j \in [n], \ S \subseteq [n] \,.
\end{align*}
We upper bound the optimal solution by finding one feasible dual solution.
Taking the dual, we obtain
    \begin{align*}
        \min \qquad &\lambda \,,  \\
        \textrm{s.t.} \qquad & \beta_{iS} + \sum_{j \neq i \mid S \in B_i \cap B_j} \eta_{ij} \ge 1\,,
        \ \ \forall\, i, \ S \in B_i\,;  \\
        & \beta_{iS} \ge \eta_{ij}\,, \ \ \forall\, i \neq j, \ S \in B_j\,; \\
        &  \sum_{i} \beta_{iS} \le \lambda\,,
        \ \ \forall\, S \subseteq [n]\,; \\
        & \eta_{ij}, \beta_{iS}, \lambda \ge 0\,, \ \ \forall \, i,j \in [n], \ S \subseteq [n] \,. 
    \end{align*}
We now set the dual solution as follows: Let $\gamma = \frac{1}{\sqrt{n}}$. Set $\eta_{ij} = \gamma$ for all $i \neq j$. For $S \notin B_i$, set $\beta_{iS} = \gamma$. For $S \in B_i$ with $|S| = q+1$, set 
$$ \beta_{iS} = \gamma + \max(0, 1 - q \cdot \gamma).$$
This setting satisfies the first and second constraints of the dual. Note now that for any $S$ with $|S| = q+1$, we have
$$ \sum_i \beta_{iS} \le n \cdot \gamma + (q+1) \cdot \max(0, 1 - q \cdot \gamma) \le 2 \sqrt{n} + 1.$$
Therefore, the dual solution has value $O(\sqrt{n})$, which by weak duality upper bounds the optimal primal solution. This in turn upper bounds $\frac{C(\ni)}{SW(\ni)}$, completing the proof.
\end{proof}

\section{Lower Bound: Proof of Theorem~\ref{thm:main11}}
\label{sec:lb1}
We complement Theorem~\ref{thm:main1} by showing a lower bound of $\Omega(\sqrt{n})$ on $\alpha(n,\cdot)$ for $\CCS$. In other words, there is an instance  $\ni = \langle [n], [m], V \rangle$, where for any allocation, some agent's utility is $O(1/\sqrt{n})$ times its cake-cutting share. By Lemma~\ref{lem:order}, this will show a corresponding lower bound for $\EFS$ and the intermediate notion of $\EF$.

Towards this end, all we need to show is an instance $\ni = \langle [n], [m], V \rangle$ with a corresponding lower bound on $\frac{C(\ni)}{SW(\ni)}$ where $C(\ni) = \sum_{i=1}^n \CCS_i$. To see this, note that for any allocation, the sum of agents' utilities is at most $SW(\ni)$. If $C(\ni) \ge \alpha \cdot SW(\ni)$, then there must be some agent whose $\CCS_i$ is at least $\alpha$ times larger than its utility.

Our instance has binary utilities with $n = q^2 + q + 1$ agents and $m = 2n - q - 1$ items, where $q$ is a large prime.
We first have a set $U$ of $n-q - 1$ items where $v_{ik} = 1$ for all agents $i \in [n]$ and $k \in U$.
The remaining set $C$ of $n$ items is induced by a finite projective plane, where for any agent $i \in [n]$,
there are exactly $q+1$ items $k \in C$ with value $v_{ik} = 1$.
In addition, for any distinct agents $i$ and $j$, there is exactly one item
$k \in C$ such that $v_{ik} = v_{jk} = 1$.
Finally, for any item $k \in C$, there are exactly $q+1$ agents $i \in [n]$ such that $v_{ik} = 1$.
A valuation system satisfying the above properties can be constructed by considering the finite projective plane of order $q$, and interpreting lines as agents and points as items; see~\cite{dembowski_1968}.

We start with the LP below capturing the ratio $\frac{C(\ni)}{SW(\ni)}$.
This is written analogously to the LP in Section~\ref{sec:step3}, but using the $\CCS$ LP (Lemma~\ref{lem:ccs_compute}) instead of the $\EFS$ one:
\begin{align}
    \max \qquad & \sum_{i \in [n]} \sum_{S \in B_i} v_S \cdot x_{iS} \notag \,, \\
    \textrm{s.t.} \qquad & \sum_{S \in B_i \cap B_j} v_S \cdot x_{iS} \le \frac{\sum_{S \in B_j} v_S}{n} \,,
    \ \ \forall\, i \neq j  \label{eq:LP-plane-cons} \,;  \\
    & 0 \le x_{iS} \le 1 \,, \ \  \forall \, i \in [n], \ S \in A_i \notag \,.
\end{align}

We need one feasible solution to the above LP to exhibit a lower bound.
We first calculate the values $v_S$ for all sets $S$.
For all items $k$, let $L_k = \{i \in [n] \mid v_{ik} = 1\}$ be the set of agents who value item $k$.
Note that for all $k \in C$, we have $|L_k| = q+1$ and $v_{L_k} = 1/m$.
It is also easy to verify that $v_{[n]} = (n-q - 1)/m$, and $v_{S'} = 0$ for all other sets $S'$.
Together, this implies that $\sum_{S \in B_i} v_S = n/m$ for any agent $i \in [n]$.

To find a feasible solution, we set $x_{iL_{k}} = 1$ for all agents $i \in [n]$ and items $k \in C$.
All other variables are set to $0$.
We now analyze constraint~\ref{eq:LP-plane-cons}.
For any distinct agents $i$ and $j$, there is exactly one item $k \in C$ such that $i \in L_k$ and $j \in L_k$, so the left-hand side of constraint~\ref{eq:LP-plane-cons} evaluates to $1/m$.
The right-hand side of constraint~\ref{eq:LP-plane-cons} also evaluates to $1/m$, which implies the constraints are satisfied.
Analyzing the objective function, we have
\[
    \frac{C(\ni)}{SW(\ni)} \ge \sum_{i \in [n]} \sum_{S \in B_i} v_S \cdot x_{iS}
    = \sum_{i \in [n]} \frac{\lvert \{k \in C \mid i \in L_k\} \rvert}{m}
    = \frac{n(q+1)}{m} = \Theta(\sqrt{n})
    .\]
This completes the proof of Theorem~\ref{thm:main11}.


%

\section{$\alpha(\cdot,m)$ for Cake Cutting Share: Proof of Theorem~\ref{thm:main2}}
\label{sec:main2}
In this section, we prove Theorem~\ref{thm:main2} for the cake-cutting shares $\{\CCS_i\}$. In particular, we show sub-linear bounds on $\alpha(\cdot, m)$ as a function of $m$, the number of items.

\subsection{Upper Bound on $\alpha(\cdot, m)$} 
Unlike the dual fitting approach described above, we will present an explicit algorithm that achieves $\alpha(\cdot,m) = O( m^{2/3})$.

Given an instance $\ni = \langle [n],[m],V \rangle$, let $\gamma_i = u_i([m]) = \sum_{k \in [m]} v_{ik}$.
We create a new instance $\ni' = \langle [n],[m],V' \rangle$, where $v'_{ik} = v_{ik}/\gamma_i$ for all $i \in [n]$ and $k \in [m]$.
Recall from Lemma~\ref{lem:ccs_compute} that $\CCS'_i$ is the solution to the following LP:

\begin{align*}
    \max \qquad &\sum_{k \in [m]} v'_{ik} \cdot x_k \,,  \\
    \textrm{s.t.} \qquad &\sum_{k \in [m]} v'_{jk} \cdot x_k  \le \frac{1}{n} \,, \ \ \forall\, j \in ([n] \setminus \{i\})\,;\\
    & 0 \le x_k \le 1 \,, \ \ \forall\, k \in [m]  \,.
\end{align*}


\paragraph{Step 1: Large-valued Items.} For every item $k$, let $w_k=\max_i v'_{ik}$ denote the largest value of any agent for this item, and let $g_k$ be any $i$ s.t. $v'_{ik}=w_k$. Our algorithm uses two parameters $a,b \in (0,1)$ that we choose later.  For each agent, we separate the items into those with ``large'' value relative to the corresponding $w_k$ and the rest.   Let $S_i=\{k\mid v'_{i k}\ge a \cdot w_k\}$ be the set of items with large values for agent $i$. 

\paragraph{Step 2: Minimal Cover} We next attempt to cover the items by the sets $\{S_i\}$ so that each set  $S_i$ covers at least $b \cdot m$ distinct items. This is achieved by Algorithm~\ref{alg1}. Notice that when the algorithm terminates, $|T|\le 1/b$ since there are only $m$ items in total, and each $S_i$ disjointly covers at least $b \cdot m$ items.

\begin{algorithm}[H]
\begin{algorithmic}[i]
  \STATE $T=\varnothing$
  \FOR{$i$ in $[n]$}    
    \IF{$|S_i\backslash (\bigcup_{j\in T}\tilde{S_j})|\ge b \cdot m$}
        \STATE $\tilde{S_i}=S_i\backslash (\bigcup_{j\in T}\tilde{S_j})$
        \STATE $T=T\cup \{i\}$
    \ENDIF
  \ENDFOR
  \end{algorithmic}
    \caption{Constructing  the Minimal Cover.}
    \label{alg1} 
\end{algorithm}

\paragraph{Step 3: Allocation.}
The final allocation has three parts:

\begin{enumerate}
    \item  For each item $k$, set $x_{g_kk} = \frac{1}{3}$, {\em i.e.}, allocate to the highest-valued agent for this item.
    \item For each $i \in T$, and $k \in \tilde{S_i}$, add $\frac{1}{3}$ to $x_{ik}$, {\em i.e.}, allocate to the items in the cover.
    \item Finally, allocate the proportional solution, that is, for each $i \in [n], k \in [m]$, add $\frac{1}{3n}$ to $x_{ik}$.
\end{enumerate}

\paragraph{Analysis.} 
The total utility of agent $i$ in the above algorithm is clearly
$$ALG_i=\frac{1}{3}\left( \frac{1}{n}+\sum_{k \mid i=g_k}v'_{g_kk}+\sum_{k\in \tilde{S_i}}v'_{ik}\right).$$

We now bound $\CCS'_i$ to align with the above expression. For this, consider the LP above, and let $x^*_{k}$ denote the value of $x_k$ in the optimal solution corresponding to $\CCS'_i$. Therefore, $\CCS'_i = \sum_k v'_{ik} x^*_{k}$.  We account for this quantity by splitting the items into three groups.

\begin{enumerate}
\item We first consider the items $k$ for which $i = g_k$ and the items $\tilde{S_i}$ if this set is non-empty. Let the union of these items be $B_i$. Since $x^*_k \le 1$, the total contribution to $\CCS'_i$ is at most
$ \sum_{k \in B_i} v'_{ik}.$
We ignore these items from further consideration, noting that their contribution to $\CCS'_i$ is within a factor of $3$ of their corresponding contribution to $ALG_i$. 
\item Next consider the items with large value that are covered by other sets in Algorithm~\ref{alg1}. In other words the set $H_i = S_i\cap (\bigcup_{j\in T, j\not= i}\tilde{S_j})$.  Notice that the number of agents in this cover is at most $1/b$. For any item $k \in H_i$, suppose $k$ is covered by the set $\tilde{S_j}$. Since $v'_{ik}, v'_{jk} \in [a \cdot w_k, w_k]$, we have $v'_{ik} \le v'_{jk} / a$.  By the LP constraint for $\CCS_i$, we have $\sum_{k \in H_i \cap \tilde{S_j}} v'_{jk} x^*_k \le 1/n$. This means $\sum_{k \in H_i \cap \tilde{S_j}} v'_{ik} x^*_k \le 1/(an).$ Since $|T| \le 1/b$, this means 
$$\sum_{k \in H_i} v'_{ik} x^*_k \le \frac{1}{a\cdot b \cdot n}.$$
\item Next consider the large value items that are not covered by Algorithm~\ref{alg1}, that is, the set $J_i = S_i\backslash (\bigcup_{j\in T_2}\tilde{S_j})$. Notice that if this set is non-empty, then $\tilde{S_i} = \emptyset$ and further, $|J_i| \le b \cdot m$. For any $k \in J_i$, consider agent $g_k$ and its LP constraint in $\CCS'_i$. By assumption, $i \neq g_k$. We have $v'_{g_k k} x^*_k \le 1/n $. Since $v'_{ik} \le v'_{g_k k} = w_k$, we have $v'_{i k} x^*_k \le 1/n $. Summing over all items in $J_i$, we have 
$$\sum_{k \in J_i} v'_{ik} x^*_k \le \frac{b\cdot m}{n}.$$ 
\item Lastly, consider the small-value items, that is, the set $\overline{S_i}$. Since each item $k \in \overline{S_i}$ has value $v'_{ik}\le a \cdot w_k$, and since $v'_{g_k k} x^*_k \le 1/n  $, we have $ v'_{i k} x^*_k \le a/n$. Note that $i \neq g_k$. As there are at most $m$ such items, we have 
$$ \sum_{k \in \overline{S_i}} v'_{ik} x^*_k \le\frac{a \cdot m}{n}.$$
\end{enumerate}

Adding up the above contributions, we have:
$$\CCS'_i \le \left(\sum_{k:i=g_k}v'_{g_kk}+\sum_{k\in \tilde{S_i}}v'_{ik} \right) +\frac{1}{n}\left(\frac{1}{a \cdot b}+(a +b) \cdot m\right).$$

Setting $a=b=m^{-\frac{1}{3}}$, this algorithm achieves $\alpha(\cdot,m) \le 3m^{\frac{2}{3}}$. The factor of $3$ can be improved to $(1+\epsilon)$ for any constant $\epsilon > 0$ by slightly modifying the algorithm, and we omit the details.

\subsection{Lower Bound on $\alpha(\cdot,m)$}
Note that the lower bound construction in Section~\ref{sec:lb1} had $m = \Theta(n)$ items.
This implies that $\alpha(\cdot,m) = \Omega(\sqrt{m})$, completing the proof of Theorem~\ref{thm:main2}.

\section{Approximation Bounds for $\pEFS$: Proof of Theorem~\ref{thm:main3}}
\label{app:partial}

\subsection{Upper Bound on Approximation}
Suppose that each agent $i$ does not know the valuations of agents in subset $W_i$, where $|W_i| = \lfloor \frac{n-1}{\Delta}\rfloor$. In this part, we assume agent $i$ chooses the best possible $W_i$ to maximize $\pEFS_i(\ni,W_i)$, and show an upper bound for this setting.  This implies an upper bound for $\pEFS_i(\ni)$. Let $Z_i = W_i \cup \{i\}$ and denote $|Z_i| = Z$ for all $i \in [n]$. Note that we assume all agents in $Z_i$ have the same utility function as agent $i$.

First note that $\pEFS_i$ for a set $W_i$ is the solution to the following LP:

\begin{align}
        \max \qquad &\sum_{k \in [m]} v_{ik} \cdot x_{ik}   \,,  \label{eq:LP-pEFSval-obj} \\
        \textrm{s.t.} \qquad &\sum_{k \in [m]} v_{jk} \cdot x_{ik}  \le \sum_{k \in [m]} v_{jk} \cdot x_{jk}\,, \ \ \forall\, j \in [n] \setminus Z_i \,;  \\
        & \sum_{j \in [n] \setminus Z_i} x_{jk} + |Z_i| \cdot x_{ik} \le 1 \,, \ \ \forall\, k \in [m] \,; \notag \\
        &  x_{jk} \ge 0 \,, \ \  \forall\, j \in [n] \setminus W_i, \ k \in [m]  \,. \notag
\end{align}


As in Section~\ref{sec:main1}, we reduce the problem of upper bounding the worst-case approximation ratio to the case with binary valuations; Theorems~\ref{thm1} and~\ref{thm2} apply as is. Proceeding as in the proof of Theorem~\ref{thm3}, the analogous LP formulation using the formulation for $\pEFS_i$ from above is: 
\begin{align*}
    \max \qquad & \sum_{i \in [n]} \sum_{S \in B_i}  v_S \cdot x_{iiS} \,, \\
    \textrm{s.t.} \qquad & \sum_{S \in B_i \cap B_j} v_S \cdot x_{iiS} \le  \sum_{S \in B_j} v_S \cdot x_{ijS} \,,  \ \  \forall \, i \in [n], \ j \in [n] \setminus Z_i \,; \\
    & \sum_{j \mid S \in B_j, j \notin Z_i} x_{ijS} + |Z_i| \cdot x_{iiS} \le 1 \,,   \ \  \forall\, i \in [n], \ S \subseteq [n] \,; \\
    & x_{ijS} \ge 0 \,,  \ \  \forall \, i,j \in [n], \ S \subseteq [n] \,.
\end{align*}
Replacing $x_{ijS} \cdot v_S$ by $z_{ijS}$ and introducing $\{v_S\}$ as variables we have: 
\begin{align*}
    \max \qquad & \sum_{i \in [n]} \sum_{S \in B_i} z_{iiS} \,, \\
    \textrm{s.t.} \qquad & \sum_{S \in B_i \cap B_j} z_{iiS} \le  \sum_{S \in B_j} z_{ijS} \,,  \ \  \forall \, i \in [n], \ j \in [n] \setminus Z_i \,; \\
    & \sum_{j \mid S \in B_j, j \notin Z_i} z_{ijS} + |Z_i| \cdot z_{iiS} \le v_S \,,   \ \  \forall\, i \in [n], \ S \subseteq [n] \,; \\
    & \sum_{S} v_S \le 1 \,;\\
    & z_{ijS} \ge 0 \,,  \ \  \forall \, i,j \in [n], \ S \subseteq [n] \,.    
\end{align*}

Taking the dual, we get:
\begin{align*}
    \min \qquad &\lambda \,,  \\
    \textrm{s.t.} \qquad & |Z_i| \cdot \beta_{iS} + \sum_{j \mid j \in [n] \setminus Z_i, S \in B_i \cap B_j} \eta_{ij} \ge 1\,,
    \ \ \forall\, i, \ S \in B_i\,; \tag{constraints for $z_{iiS}$}  \\
    & \beta_{iS} \ge \eta_{ij}\,, \ \ \forall\, i \in [n]\,, \ j \in [n] \setminus Z_i \,, \ S \in B_j\,; \tag{constraints for $z_{ijS}$}\\
    &  \sum_{i} \beta_{iS} \le \lambda\,,
    \ \ \forall\, S \subseteq [n]\,; \tag{constraints for $v_S$}\\
    & \eta_{ij}, \beta_{iS}, \lambda \ge 0\,, \ \ \forall \, i,j \in [n], \ S \subseteq [n] \,. 
\end{align*}

Set all $\eta_{ij} = \gamma = \frac{1}{\sqrt{nZ}}$. If $S\notin  B_i$, set $\beta_{iS} = \gamma$. If $S \in B_i$, let $q_{iS} = |S\setminus Z_i|$ and set $\beta_{iS} = \gamma + \max\{\frac{1-\gamma \cdot q_{iS}}{|Z_i|}, 0\}$. The first two constraints in the above LP directly hold. We only need an upper bound on $\lambda$ to make the third constraint hold. 

We have
\begin{align*}
    \sum_{i} \beta_{iS} &\le n \cdot \gamma + \sum_{i \in S} \max \left\{0, \frac{1 - \gamma \cdot q_{iS}}{|Z_i|}\right\} \\
    &\le n \cdot \gamma + \sum_{i \in S} \max \left\{0, \frac{1 - \gamma \cdot (|S| - |Z_i|)}{|Z_i|}\right\} \\
    &\le  n \cdot \gamma + \max \left\{0, \frac{|S| - \gamma \cdot (|S| - Z) \cdot |S|}{Z}\right\} \\
    &\le  n \cdot \gamma + \max \left\{0, \frac{(\gamma \cdot Z + 1)^2}{4\gamma \cdot Z}\right\} \\
    & \le 2 \sqrt{n/Z} \tag{Since $\gamma = \frac{1}{\sqrt{n \cdot Z}}$}.
\end{align*}
Therefore, $\sum_i \pEFS_i$ is at most $2\sqrt{n/Z}$ times the social welfare. This completes the proof of the upper bound of $O(\sqrt{\Delta})$ on the approximation ratio.

\subsection{Lower Bound on Approximation}
Recall that $Z = |W| + 1 = 1 + \lfloor \frac{n-1}{\Delta} \rfloor$.
Our instance will have $n$ agents and $\binom{n}{\ell}$ items where $\ell = \frac{\sqrt{nZ}}{2}$.
We simply have one item $S$ for each subset of $\ell$ agents, with $v_{iS} = 1$ for $i \in S$  and $v_{iS} = 0$ otherwise.

Focus on some $i$ and fix any set $W_i$ of size $Z-1$. We compute a lower bound on $\pEFS_i$ by showing an envy-free allocation. We separate the set of items into two parts -- those contained in $B_i$ and the rest (call this set $T$). There are $\binom{n-1}{\ell-1}$ items in set $B_i$; for $S \in B_i$, set $x_{jS} = \frac{1}{Z}$ for $j \in Z_i$, and $x_{jS} = 0$ otherwise. For $S \in T$, set $x_{jS} = \frac{1}{\ell}$ if $j \in S$, and $x_{jS} = 0$ otherwise. This allocation is feasible in that for any item, the total allocated is at most $1$.

Clearly the allocation for $j \in Z_i$ is the same as the allocation for $i$. For $j \notin Z_i$, we have $|B_j \cap B_i| = {n-2 \choose \ell-2}$ and $|B_j \setminus B_i| = {n-1 \choose \ell-1} - {n-2 \choose \ell-2} = \frac{n-\ell}{\ell - 1} \cdot {n-2 \choose \ell-2}$. We now show that $j$ does not envy $i$'s allocation. We have 
$$ u_j(A_j) = \frac{|B_j \setminus B_i|}{\ell} = \frac{1}{\ell} \cdot \frac{n-\ell}{\ell - 1} \cdot {n-2 \choose \ell-2} \ge \frac{4(n-\ell)}{n \cdot Z} \cdot {n-2 \choose \ell-2} > \frac{1}{Z} \cdot {n-2 \choose \ell-2} = u_j(A_i), $$
where the final inequality follows since $3n > 2n \ge 4 \ell$. Using this allocation, 
$$\pEFS_i \ge u_i(A_i) \ge \frac{1}{Z} \cdot {n-1 \choose \ell-1}.$$

Considering all $i$ together, the social welfare is simply the number of items, which is ${n \choose \ell}$, while 
$$\sum_i \pEFS_i \ge \frac{n}{Z} \cdot {n-1 \choose \ell-1} = \frac{n}{Z} \cdot \frac{\ell}{n} {n \choose \ell} = \sqrt{\frac{n}{4Z}} \cdot  {n \choose \ell}.$$
Therefore, there exists an agent whose $\pEFS_i$ cannot be approximated below a $\Omega\big(\sqrt{\Delta} \big)$ factor, completing the proof of the lower bound.

\section{Empirical Results}
\label{sec:expt}
Given the $\Theta(\sqrt{n})$ approximation bounds for $\CCS$ and $\EFS$, one may wonder which shares are more realistic to target in practice. We therefore empirically measure the approximation ratio $\theta$ to $\CCS$ and $\EFS$ achievable on simulated and real data. In the same vein, we measure how the approximation ratio for $\pEFS$ (see Definition~\ref{def:pefs}) changes with the parameter $\Delta$. (Recall the definition of $\theta$ from Section~\ref{sec:approx}.) 


In the simulated model, we set $n=25$ and $m = 3n = 75$ for all instances. We randomly generate $200$ fair division instances and compute the approximation guarantee given by each fairness notion using the LP in 
Eq~(\ref{lp:approx}). We use the setting in \citet{Caragiannis2019nash} and set $\vec{v_i}$ to be a uniformly random integer $m$-partition that sums to $1000$. 

The real data uses the Yahoo A1 dataset \cite{Yahoo} which contains bid information for advertisers (agents) on a set of ads (items). We generate $200$ instances by selecting a random set of $m = 20$ ads for each instance, and taking the $n=20$ advertisers with the most bids on those ads. Agent $i$'s utility for an item $j$ is set to be advertiser $i$'s bid for ad $j$.

The approximation ratios $\theta$ for $\PROP$, $\CCS$, and $\EFS$ are plotted in Figures~\ref{fig:uniform_sum_to_1000} and \ref{fig:yahoo}, for the simulated and real datasets respectively.  In our instances, the $\theta$ for $\CCS$ is approximately one, which suggests that $\CCS$ is a reasonable notion to target in practice. In contrast, the $\theta$ for $\PROP$ is much higher than one, suggesting it underestimates shares, while those for $\EFS$ are much lower than one, which suggests it cannot be simultaneously achieved for all agents. For the Yahoo dataset, though $\CCS$ and $\EFS$ seem equally reasonable, $\CCS$ yields feasible shares while $\EFS$ does not.

For computing $\pEFS$, we use the same setup as before. For each instance, we sample $20$ sets $W_i$ of size $\lfloor \frac{n-1}{\Delta}\rfloor$ to estimate the expectation in Definition~\ref{def:pefs}. We plot the average $\theta$ over the $200$ instances as a function of $\Delta$ for $\pEFS$ in Figures~\ref{fig:uniform2} and~\ref{fig:yahoo2}, for the simulated and real datasets respectively. Note that the approximation ratio smoothly decreases with $\Delta$. For both the datasets, in order for the ratio to be approximately $1$,  the average size of $W_i$ of agents whose value is unknown to $i$ is approximately $6$, which is between a quarter and a third of the number of agents. 


\begin{figure*}[hbt!]
    \begin{subfigure}[t]{0.245\textwidth}
        \centering
        \includegraphics[width=\linewidth]{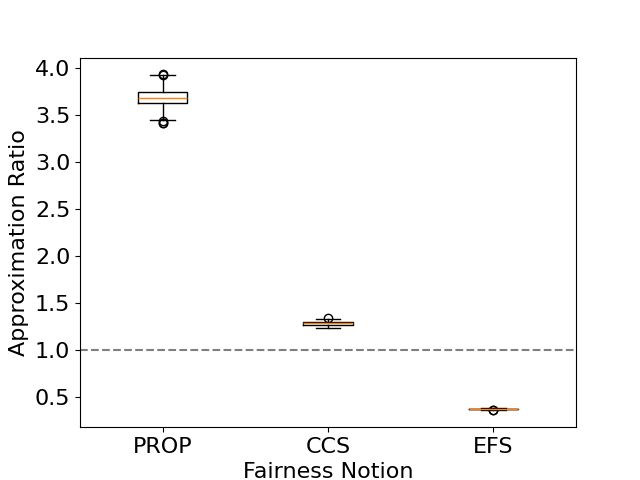}
        \caption{Uniform Value}
        \label{fig:uniform_sum_to_1000}
    \end{subfigure}
    \begin{subfigure}[t]{0.245\textwidth}
        \centering
        \includegraphics[width=\linewidth]{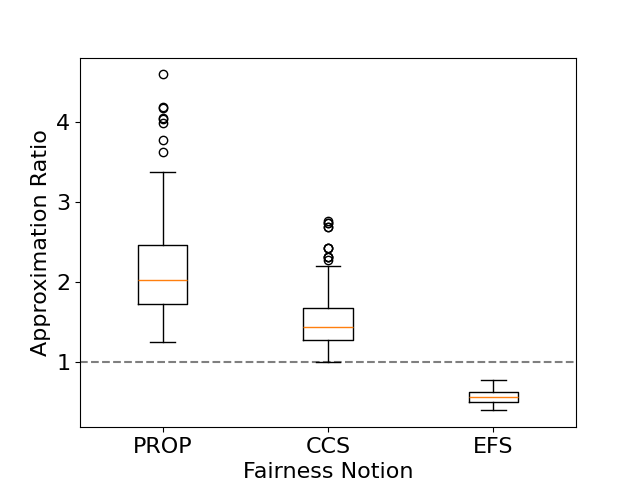}
        \caption{Yahoo Data}
        \label{fig:yahoo}
     \end{subfigure}
    \begin{subfigure}[t]{.245\textwidth}
        \centering
        \includegraphics[width=\linewidth]{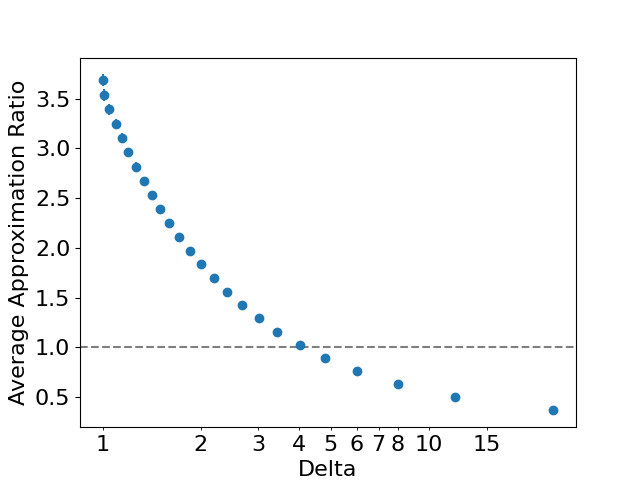}
        \caption{Uniform Value: $\pEFS$}
        \label{fig:uniform2}
    \end{subfigure}
    \begin{subfigure}[t]{0.245\textwidth}
        \centering
        \includegraphics[width=\linewidth]{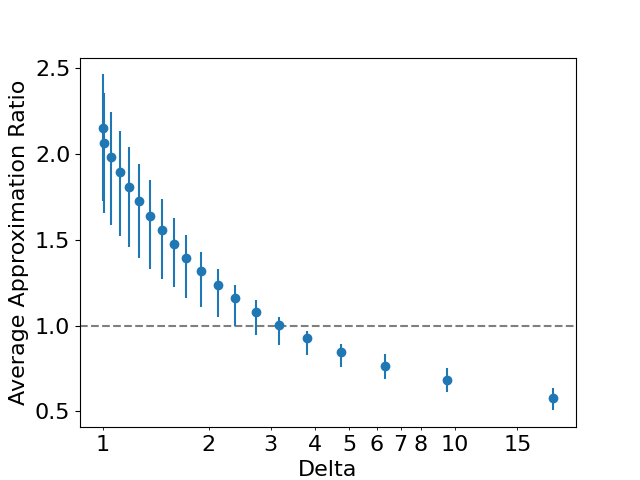}
        \caption{Yahoo Data: $\pEFS$}
        \label{fig:yahoo2}
    \end{subfigure}
    \caption{Figures (a), (b): Box plot of approximation Ratios given by $\PROP$, $\CCS$, and $\EFS$. Figures (c), (d): Plot of mean and inter-quartile range of approximation ratio as a function of $\Delta$ for $\pEFS$.}
\end{figure*}

\paragraph{Additional Simulated Models.} We now present additional empirical results for the approximation ratio of $\CCS$ and $\EFS$. As before, we set $n=25$ and $m = 3n = 75$ for all instances. For each model of agent utilities, we randomly generate $200$ fair division instances and compute the approximation guarantee given by each fairness notion using the LP in 
Eq~(\ref{lp:approx}).

In the Bernoulli model, each agent's utility for each item is drawn from a Bernoulli distribution with parameter $p=.5$, i.e., for all $i \in [n]$ and $k \in [m]$, we set $v_{ik} = 1$ with probability $0.5$ and $v_{ik} = 0$ otherwise. 
In the Intrinsic Value model, we simulate a situation where every item has an intrinsic value and each agent has an additional value for each item which captures its personal preferences. Specifically, for each item $k \in [m]$, we first sample $\alpha_k \sim U(0, 1)$, which captures the intrinsic value of item $k$.
Then for all $i \in [n]$ and $k \in [m]$, we sample $\beta_{ik} \sim U(0, 0.3)$, which captures each agent's personal preferences.
We set $v_{ik} = \alpha_k + \beta_{ik}$.

The $\theta$ for $\PROP$, $\CCS$, and $\EFS$ are shown in Figures~\ref{fig:bernoulli} and~\ref{fig:intrinsic_value} respectively. As before, the $\CCS$ share has approximation ratio closest to $1$ and is feasible to achieve simultaneously for all agents. This shows our conclusions are fairly robust to agent utility models.

\begin{figure*}[htbp]
    \centering
    \begin{subfigure}[t]{.25\textwidth}
        \centering
        \includegraphics[width=\linewidth]{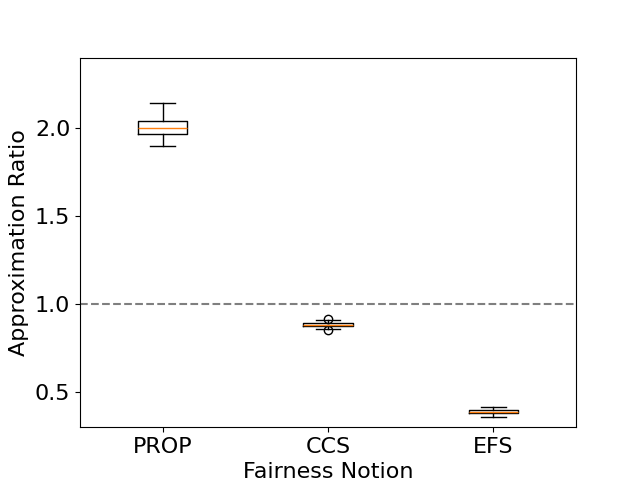}
        \caption{Bernoulli Value.}
        \label{fig:bernoulli}
    \end{subfigure}
    \qquad
    \qquad
    \begin{subfigure}[t]{0.25\textwidth}
        \centering
        \includegraphics[width=\linewidth]{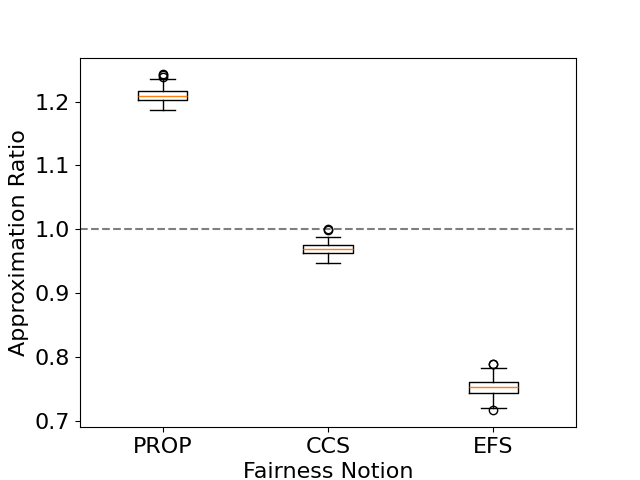}
        \caption{Intrinsic Value.}
        \label{fig:intrinsic_value}
    \end{subfigure}
    
    \caption{Approximation Ratios given by $\PROP$, $\CCS$, and $\EFS$ for Bernoulli distribution and Intrinsic Value.}
\end{figure*}
\section{Conclusion}
In this paper, we have presented fair share concepts that are appealing for three reasons. First, unlike proportionality, they are sensitive to the utility functions of other agents, so their value is more representative of the true fair share achievable by an agent. Second, on any instance, these shares, as well as the allocation that achieves the optimal approximation to them, can be efficiently computed. Finally, there are allocations that achieve non-trivial worst-case approximation (that is, an approximation factor of $o(n)$) to these shares. This makes these shares a meaningful concept to target in practice, and we support the final claim via experimental results.

Analogous to \citet{Babaioff}, our work begins exploring notions of fair shares that go beyond proportionality. The main open question stemming from our work is whether there are other notions of fair shares that are intermediate between being oblivious to other agents' utility functions (like proportionality), and those that fully consider their utilities (like our shares). 
It would also be interesting to develop analogs of our concepts for more general monotone utility functions, as well as for indivisible goods, and study their approximation bounds.
\section*{Acknowledgments}
We thank Noga Alon for suggesting an improved lower bound construction in Section~\ref{sec:lb1}.

\bibliographystyle{plainnat}
\bibliography{refs}

\end{document}